\documentclass{llncs}

\usepackage{macros}
\usepackage{url}
\usepackage{framed}
\usepackage{graphicx}
\usepackage{physics}

\DeclareMathOperator{\disc}{disc}
\DeclareMathOperator{\Prob}{Prob}

\newtheorem{Lemma}{Lemma}
\newtheorem{Proposition}{Proposition}
\usepackage{algpseudocode}
\usepackage{algorithmicx}
\usepackage{algorithm}

\author{Hao Chen\inst{1} \and Kristin Lauter\inst{2} \and Katherine E. Stange \inst{3}}

\institute{University of Washington, Seattle, USA \\
 \email{chenh123uw.edu} \\
 \and
Microsoft Research, Redmond, USA \\
 \email{klauter@microsoft.com} \\
\and
University of Colorado, Boulder, USA \\
\email{kstange@math.colorado.edu}}

\begin{document}

\iffalse
Summary:
H = subgroup of (Z/1468005Z)^* of order 4000 generated by [198892, 978671, 431521, 1083139]
n = 144
q = 461
sigma_0 = 1.00000000000000
success ? True
number of samples = 18440
Total Time = 7692.044723
Attacktime  = 1490.17313

Summary:
H = subgroup of (Z/1468005Z)^* of order 4000 generated by [198892, 978671, 431521, 1083139]
n = 144
q = 953
sigma_0 = 1.00000000000000
success ? True
number of samples = 38120
Total Time = 23144.022411
Attacktime  = 6871.661452
\fi

\title{Security considerations for Galois non-dual RLWE families}

\maketitle

\begin{abstract}
We explore further the hardness of the non-dual discrete variant of the Ring-LWE problem for various number rings, give improved attacks for certain rings satisfying some additional assumptions, construct a new family of vulnerable Galois number fields, and apply some number theoretic results on Gauss sums to deduce the likely failure of these attacks for 2-power cyclotomic rings and unramified moduli.
\end{abstract}

\section{Introduction}

Lattice-based cryptography was introduced in the mid 1990s in two different forms, independently by Ajtai-Dwork~\cite{ajtai1997public} and Hoffstein-Pipher-Silverman~\cite{hoffstein2008introduction}.  Thanks to the work of Stehl\'e-Steinfeld~\cite{stehle2011making}, we now understand the NTRU cryptosystem introduced by Hoffstein-Pipher-Silverman to be a variant of a cryptosystem which has security reductions to the Ring Learning With Errors (RLWE) problem. The RLWE problem was introduced in~\cite{lyubashevsky2013ideal} as a version of the LWE problem~\cite{regev2009lattices}: both problems have reductions to hard lattice problems and thus are interesting for practical applications in cryptography.  RLWE depends on a number ring $R$, a modulus $q$, and an error distribution.  As such, it has added structure (the ring), which allows for greater efficiency, but also in some cases additional attacks.

The hardness of RLWE is crucial to cryptography, in particular as the basis of numerous homomorphic encryption schemes~\cite{bos2013improved,brakerski2012fully,brakerski2012leveled,brakerski2011fully,brakerski2014efficient,lopez2012fly,stehle2011making}.  One main theoretical result in this direction is the security reduction theorem in \cite{lyubashevsky2013ideal}, which reduces certain GapSVP problems in ideal lattices over $R$ to RLWE, when the RLWE error distribution is sufficiently large and of a prescribed form.  Although so far in practical cryptographic applications only cyclotomic rings are used, it is important to study the hardness of RLWE for general number rings, moduli and error distributions, so as to understand the boundaries of security in the parameter space.  Recently, new attacks on the so-called \emph{non-dual discrete} variant of the RLWE problem for certain number rings, error distributions, and special moduli were introduced~\cite{cryptoeprint:2016:240,castryck2016provably,cryptoeprint:2015:971,eisentrager2014weak,elos2015weak}.  The RLWE problem reduces to its discrete variant; and the non-dual RLWE problem is equivalent to the dual problem up to a change in the error distribution, so that non-dual RLWE may be viewed simply as a certain choice of error distribution in the parameter space of RLWE.  The term \emph{RLWE} is sometimes reserved for spherical Gaussian distributions.

This paper is an extension of \cite{cryptoeprint:2015:971}, and here we explore further the hardness of the non-dual discrete variant of the  RLWE problem for various number rings.  We:
\begin{enumerate}
        \item construct a new family of vulnerable Galois number fields,
        \item improve the runtime of the attacks for certain rings satisfying some additional assumptions, and 
        \item apply some number theoretic results on Gauss sums to deduce the likely failure of these attacks for 2-power cyclotomic rings.
\end{enumerate}

In cryptographic applications, it is most efficient to sample the error distribution coordinate-wise according to a polynomial basis for the ring.  
For $2$-power cyclotomic rings, which are monogenic with a well-behaved power basis, it is justified to sample the RLWE error distribution directly in the polynomial basis for the ring, according to results in~\cite{brakerski2011fully,eisentrager2014weak,lyubashevsky2013ideal}, where this error distribution choice is called Polynomial Learning With Errors (PLWE).  Precisely, the PLWE (polynomial error), RLWE (meaning a spherical Gaussian), and non-dual RLWE problems are equivalent up to a scaling and rotation of the error distribution for $2$-power cyclotomic fields.  However, in general number rings the error distribution may be distorted by a general linear transformation when moving from one problem to another \cite{elos2015weak}.   For certain choices of ring and modulus, efficient attacks on PLWE were presented in~\cite{eisentrager2014weak}.  In~\cite{elos2015weak}, these attacks were extended to apply to the decision version of the non-dual RLWE problem in certain rings, and in~\cite{castryck2016provably,cryptoeprint:2015:971}, attacks on the search version of the RLWE problem for certain choices of ring and modulus were presented.
%So it is important to study the hardness of the both PLWE and non-dual RLWE problems and the relationship between the two problems in general rings.

\subsection{Summary of contributions}

\begin{itemize}
\item In Section 3, we present an improvement to the attack in  \cite[Section 4]{cryptoeprint:2015:971} and use it to dramatically cut down the runtime of the attacks on the weak instances found in \cite[Section 5]{cryptoeprint:2015:971}.

\item In Section 4, we present a new infinite family of Galois number fields vulnerable to our attack in \cite[Section 4]{cryptoeprint:2015:971}, where
the relative standard deviation parameter is allowed to grow to infinity, and we give a table of examples.

\item In Section 5, we analyze the security of 2-power cyclotomic fields with unramified moduli under our attack. We prove Theorem~\ref{cyclo secure}, which gives an upper bound on the statistical distance between an approximated non-dual RLWE error distribution, reduced modulo a prime ideal $\fq$, and the uniform distribution on $R/\fq$. We conclude that the 2-power cyclotomic rings are safe against our attack when the modulus $q$ is  unramified  with small residue degree (1 or 2), and is not too large ($q < m^2$).
\end{itemize}

%\begin{acknowledgements}
{\bf Acknowledgements} We thank Chris Peikert, Igor Shparlinski, L\'eo Ducas and Ronald Cramer for helpful discussions.
%\end{acknowledgements}

\section{Background}

\subsection{Discrete Gaussian on lattices}
Recall that a {\it lattice} in $\bR^n$ is a discrete subgroup of $\bR^n$ of rank $n$.  For $r > 0$,  let $\rho_r (x) = e^{-||x||^2/r^2}$.
%(our $\sigma$ is equal to $r/\sqrt{2\pi}$ for the parameter $r$ in \cite{lyubashevsky2013ideal}).
\begin{definition}
For a lattice $\Lambda \subset \bR^n$ and $r > 0$, the {\it discrete Gaussian distribution} on $\Lambda$ with width $r$ is:
\[
    D_{\Lambda, r}(x) = \frac{\rho_r(x)}{\sum_{y \in\Lambda} \rho_r(y)}, \, \forall x \in \Lambda.
\]
\end{definition}

\subsection{Non-dual RLWE}

A non-dual discrete RLWE instance is specified by a ring $R$, a positive integer $q$ and an error distribution $\chi$ over $R$.  Here $R$ is normally taken to be the ring of integers of some number field $K$ of degree $n$. The integer $q$, called the {\it modulus}, is often taken to be a prime number. We then fix an element $s \in R/qR$ called the {\it secret}.

Let $\iota: K \to \bR^n$ be the \emph{adjusted canonical embedding} defined as follows.  Suppose $\sigma_1, \ldots, \sigma_{r_1}, \sigma_{r_1+1}, \ldots, \sigma_n$ are the distinct embeddings of $K$, such that $\sigma_1, \cdots, \sigma_{r_1}$ are the real embeddings and $\sigma_{r_1+r_2+j} = \overline{\sigma_{r_1 + j}}$ for $1 \leq j \leq r_2$. We define $\iota: K \to \bR^n$ by
\begin{eqnarray*}
    x \mapsto (\sigma_1(x), \cdots, \sigma_{r_1}(x), \sqrt{2}\Re(\sigma_{r_1+1}(x)), \sqrt{2}\Im(\sigma_{r_1+1}(x)), \cdots,  \\ \sqrt{2} \Re(\sigma_{r_1+r_2}(x)), \sqrt{2} \Im(\sigma_{r_1+r_2}(x))).
\end{eqnarray*}

Then the {\it non-dual discrete RLWE error distribution} is the discrete Gaussian distribution $D_{\iota(R),r}$.
%or the continuous spherical Gaussian distribution with width $\sqrt{2 \pi} \sigma$ in the dual case.

\begin{definition}
Fix $R,q, r$ as above. Let $R_q$ denote the quotient ring $R/qR$. Then a non-dual RLWE sample is a pair $$(a, b = as+e) \in R_q \times R_q, $$ where the first coordinate $a$ is chosen uniformly at random in $R_q$, and $e$ is a sampled from the discrete Gaussian $D_{\iota(R), r}$, considered modulo $q$. 
\end{definition}

\begin{definition}[Non-dual Search RLWE]
        Given arbitrarily many non-dual RLWE samples, determine the secret $s$.
\end{definition}

\begin{definition}[Non-dual Decision RLWE]
        Given arbitrarily many samples in $R_q \times R_q$, which are either non-dual RLWE samples for a fixed secret $s$, or uniformly random samples, determine which.
\end{definition}

\subsection{Comparing RLWE with non-dual RLWE}

In the original work \cite{lyubashevsky2013ideal}, the RLWE problem is introduced using the dual ring $R^{\vee}$.  Specifically, for the discrete variant, $s \in R_q^\vee := R^\vee /qR^\vee$, and an RLWE sample is taken to be of the form
 $$(a, b = as+e) \in R_q \times R_q^\vee,$$ 
 where $e$ is sampled from $D_{\iota(R^\vee),r}$, then considered modulo $q$.

 If the dual ring $R^{\vee}$ is principal as a fractional ideal, i.e., $R^{\vee} = tR$, then each non-dual instance is equivalent to a dual instance, by mapping a sample $(a,b)$ to $(a,tb)$, and vice versa.  If $R^\vee$ is not principal, there are still inclusions $R^\vee \subset t_1R$ and $R \subset t_2R^\vee$, so that one can reduce dual and non-dual versions of the problem to one another.  In either case, the reduction comes at the cost of distorting the error distribution.

For the infinite family constructed in Section~\ref{section: family}, the dual ring  $R^{\vee}$ is indeed principal (see Lemma~\ref{field theory} in Section~\ref{section: family}).  Note that multiplying by this field element $t$ changes a spherical Gaussian to an elliptical Gaussian, so the two equivalent instances will have different error shapes.

Elliptical Gaussians are the most important class of error distributions for general rings, since in \cite[Theorem 4.1]{lyubashevsky2013ideal}, the reduction from hard lattice problems is to a class of RLWE problems where the distributions are elliptical Gaussians. 
Theorem~5.2 of  \cite{lyubashevsky2013ideal} provides a further security reduction to decision RLWE with spherical Gaussian errors, but it is only stated for cyclotomic rings.

\subsection{Comparing discrete and continuous errors}

Restricting now to the non-dual setup, there are still two variants of RLWE based on the form of the spherical errors: the {\it continuous} variant samples errors from spherical Gaussian on the space $K_\bR = \iota( K \otimes_\bQ \bR)$ (here we extend $\iota$ linearly), so that samples have the form
\[
 (a, b = as + e) \in R_q \times K_\bR/qR,        
\]
whereas the {\it discrete} variant samples from a discrete Gaussian $D_{\iota{R},r}$ on the lattice $R$, as defined above.

There is no known equivalence between the discrete problem and its continuous counterpart in general.  However, the continuous problem reduces to the discrete one.  Specifically, given a continuous sample $(a, b) \in R_q \times  K_\bR/qR$, one can perform a rounding on the second coordinate to get a discrete sample $(a,[b]) \in R_q \times R_q$.
However, there is no obvious map in the reverse direction.  %Is it possible that the discrete problem is strictly more difficult?

%However, it is not clear that the converse is true, and we hope this might make an interesting research topic. 
%As an example, for the vulnerable family we constructed in Section 2, the vulnerability depends on
%specific choices of modulus $q$. In particular, for some choices of $q$, the attack would not work.
%However, the recent preprint \cite{peikert2016not} demonstrated a trace distinguishing attack, which solves the decision RLWE problem for the exact problem with discrete errors replaced by continuous errors, and the attack  works regardless of the modulus $q$ (see the Appendix). These phenomena show that the comparison between the securities of discrete versus continuous RLWE is not trivial.  On the contrary,  it is an open question, and we believe that it deserves further study:  while security reductions refer to continuous RLWE, practical implementations are usually expected to be discrete.

\subsection{Search and decision RLWE problems}

Let $\fq$ be a prime ideal of $K$ lying above $q$; then the RLWE problem modulo $\fq$ means discovering $s \mod{\fq}$ from arbitrarily many RLWE samples. In \cite{lyubashevsky2013ideal} the authors gave a polynomial time reduction from search to decision for cyclotomic number fields and totally split primes, using the RLWE modulo $\fq$ as an intermediate problem. Their proof can be applied to prove a similar search-to-decision reduction for non-dual RLWE, when the underlying number field is Galois and the modulus $q$ is unramified \cite{cryptoeprint:2015:971,eisentrager2014weak}. Moreover, the search-to-decision is most efficient when the residue degree of $q$ is small.  What is important in our paper is that for the instances in Section 3 and 4, our attacks on RLWE modulo $\fq$ could be efficiently transferred to attack the search problem. 

\subsection{Comparing non-dual RLWE with PLWE for 2-power cyclotomic fields}

For cryptographic applications, it is perhaps natural to consider the PLWE error distribution on $R$: assuming the ring $R$ is monogenic, i.e., $R = \bZ[x]/(f(x))$, then a sample from the PLWE error distribution is $e = \sum_{i=0}^{n-1} e_i x^i$, where the $e_i$ are ``small errors'', sampled independently from some error distribution over $\bZ$ (e.g. a discrete Gaussian distribution). 

In general number fields, a PLWE distribution differs greatly from the non-dual RLWE distribution (see [ELOS] for an effort to quantify the distance between the two distributions using spectral norms). However, for 2-power cyclotomic fields it turns out that the two error distributions are equivalent up to a factor of $\sqrt{n}$. Since this fact is used in Section 5, we give a proof below. 

\begin{Lemma}
Let $m = 2^d$ be a power of 2 and let $R = \bZ[\zeta_m]$. Consider the PLWE error distribution on $R$, i.e. samples
$e = \sum_{i=0}^{n-1} e_i \zeta_{m}^i$, where $n = m/2$ and each $e_i$ follows the discrete Gaussian $D_{\bZ,r}$. Then this PLWE distribution is equal to the non-dual RLWE distribution $D_{\iota(R), r\sqrt{n}}$. 
\end{Lemma}

\begin{proof}
For an element $x = \sum_{i=0}^{n-1} x_i\zeta_m^i \in R$, the probability of $x$ being sampled by the PLWE distribution is proportional to $\prod_{i=0}^{n-1} \rho_r( x_i) = \prod_{i=0}^{n-1} e^{-x_i^2/r^2} =  e^{- ||x||^2/r^2}$. 
On the other hand,  one checks that $||\iota(x)|| = \sqrt{n} ||x||$.  So the above probability is proportional to $e^{- ||\iota(x)||^2/nr^2}$, which is the exactly the same for the distribution $D_{\iota(R), r\sqrt{n}}$. This completes the proof. 
\end{proof}

\subsection{Scaling factors}

As pointed out in \cite{elos2015weak}, when analyzing the non-dual RLWE error distribution, one needs to take into account the sparsity of the lattice $\iota(R)$, measured by its covolume in $\bR^n$. This covolume is equal to $|\disc(K)|^{1/2}$. In light of this, we define the scaled error width to be $$r_0 = \frac{r}{|\disc(K)|^{\frac{1}{2n}}}.$$

% Also Castryck et al. suggested that taking the error to be $\Omega(|\disc(K)^{1/n}|$ 

\subsection{Overview of attack}
\label{subsec: attack}

We briefly review the method of attack in Section 4 of \cite{cryptoeprint:2015:971}.  The basic principle of this family of attacks is to find a homomorphism 
\[
        \rho: R_q \rightarrow F
\]
to some small finite field $F$, such that the error distribution on $R_q$ is transported by $\rho$ to a non-uniform distribution on $F$.  In this case, errors can be distinguished from elements uniformly drawn from $R_q$ by a statistical test in $F$, for example, by a $\chi^2$-test.  The existence (or non-existence) of such a homomorphism depends on the parameters of the field, prime, and distribution in the setup of RLWE.  In this section, we will describe parameters under which such a map exists.

Once such a map is known, the basic method of attack on Decision RLWE is as follows:
\begin{enumerate}
        \item Apply $\rho$ to samples $(a,b)$ in $R_q \times R_q$, to obtain samples in $F \times F$.
        \item Guess the image of the secret $\rho(s)$ in $F$, calling the guess $g$.
        \item Compute the distribution of $\rho(b) - \rho(a)g$ for all the samples.  If $g = \rho(s)$, this is the image of the distribution of the errors.  Otherwise it is the image of a uniform distribution.
        \item If the image looks uniform, try another guess $g$ until all are exhausted.  If any non-uniform distribution is found, the samples are RLWE samples.  Otherwise they are not.
\end{enumerate}

Whenever $\mathfrak{q}$ is a prime ideal lying above $q$, then reduction modulo $\mathfrak{q}$ is a valid map
\[
        \rho : R_q \rightarrow R_\mathfrak{q}
\]
for the attack above.  This attack targets the RLWE modulo $\mathfrak{q}$ problem for some prime $\mathfrak{q}$ lying above $q$, and as noted above, it can be turned into an attack on the search variant of the problem, whenever $q$ is unramified and $K$ is Galois.

\subsection{Comparison to related works}

In an independent preprint (\cite{cryptoeprint:2016:240}) which appeared on eprint around the same time as our preprint, Castryck et al. also constructed an infinite family of vulnerable Galois number fields, where the error width can be taken to be $O(|\disc(K)|^{\frac{1}{n} - \epsilon})$ for any $\epsilon > 0$. The asymptotic error width they obtained is wider than in our infinite family in Section 2.  However, the method of attack is an errorless LWE linear algebra attack (based on short vectors), whereas our family is not susceptible to a linear algebra attack, and requires the novel techniques presented here and in \cite{cryptoeprint:2015:971}.

%Thank the anonymous referee for helpful suggestions and comments.

 %We remark that the runtime has excluded the time taken to reduce the
%samples to $\bF_{q^2}$, since we can precompute a basis and this is doing a dot product for each sample, the time taken
%will be $O(nq)$, which is dominated by the complexity $O(q^2)$ in the main part of the attack as long as $p = o(q)$.

% section on cosets
\section{An improved attack using cosets}

In this section, we describe an improvement to our chi-square attack on RLWE mod $\mathfrak{q}$ outlined in Section~\ref{subsec: attack} for a special case. As a result, we have an updated version of \cite[Table 1]{cryptoeprint:2015:971}, where we attacked each instance in the table in much shorter time. Note that the complexity of the previous attack in this special case is $O(nq^3)$. In contrast, our new attack has complexity $O(nq^2)$.

To clarify, the special case we consider in this section is characterised by the following assumptions (we need not be in the special family of the next section):
\begin{itemize}
\item The modulus $q$ is a prime of residue degree 2 in the number field $K$.
\item There exists a prime ideal $\fq$ above $q$ such that the map $\rho: R_q \to R_\fq$ satisfies the following property:
Let $e \in R_q$ be taken from the discrete RLWE error distribution. The probability that $\rho(e)$ lies in the prime subfield $\bF_q$ of $\bF_{q^2}$ is computationally distinguishable from $1/q$.
\end{itemize}

Granting these assumptions, we can distinguish the distribution of the ``reduced error" $\rho(e)$ from the uniform distribution on $\bF_{q^2}$. More precisely, the attack in \cite{cryptoeprint:2015:971} works exactly as we described in Section~\ref{section: family}:  with access to $\Omega(q)$ samples, one loops over all $q^2$ possible values of $\rho(s)$. It detects the correct guess $\rho(s)$ based on a chi-square test with two bins $\bF_q$ and $\bF_{q^2} \setminus \bF_q$.

The distinguishing feature of the improved attack is to loop over the cosets of $\bF_{q}$ of $\bF_{q^2}$ instead of the whole space. Fix $t_1, \cdots t_q$ to be a set of coset representatives for the additive group $\bF_{q^2}/\bF_q$. Recall that $s$ denotes the secret and $\rho: R_q \to R_\fq \cong \bF_{q^2}$ is a reduction map modulo some fixed prime ideal $\fq$ lying above $q$.  Then there exists a unique index $i$ such that $\rho(s)= s_0 + t_i$ for some $s_0 \in \bF_q$. Our improved attack will recover $s_0$ and $t_i$ separately.  \\

We start with an identity $b = as+e$, where $a,b,s,e \in \bF_{q^2}$. We will regard $s$ as fixed and $a,b,e$ as random variables, such that $a$ is uniformly distributed in $\bF_{q^2} \setminus \bF_q$ and $b$ is uniformly distributed in $\bF_{q^2}$. The reason why $a$ is not taken to be uniform will become clear later in this section.  We use a bar to denote the Frobenius automorphism, i.e.,
\[
	\bar{a} \stackrel{def}{=} a^q, \, \forall a \in \bF_{q^2}.
\]
Then $\bar{b} = \bar{a}\bar{s} + \bar{e}$. Using the identity $s = s_0 + t_i$ and subtracting,  we obtain  $\bar{b} - b - \overline{at_i} + at_i= s_0(\bar{a} - a) + \bar{e} - e$. Since $a \neq \bar{a}$, we can divide through by $\bar{a}-a$ and get
\[
\frac{\bar{b} - b - \overline{at_i} + at_i}{\bar{a}-a} = s_0 + \frac{\bar{e} - e}{\bar{a}-a}. \tag{**}
\]
Now for each $1 \leq j \leq q$, we can compute
\[
        m_j(a,b) := \frac{\bar{b} - b - \overline{at_j} + at_j}{\bar{a}-a}
\]
with access to $a$ and $b$, but without knowledge of $s$ or $s_0$.
Note that $m_j$ is in the prime field $\bF_q$ by construction.

% Remark: $m_j$ is undefined when $a  = \bar{a}$. So we need to throw away such samples.

\begin{Proposition} \label{dist of mj} For each $1 \leq j \leq q$, \\
        (1) If $j \neq i$, then $m_j(a,b)$ is uniformly distributed in $\bF_q$, for RLWE samples $(a,b)$.  \\
        (2) If $j = i$, then $m_j(a,b) = s_0 + \frac{\bar{e} - e}{\bar{a}-a}$.
\end{Proposition}

We postpone the proof of Proposition~\ref{dist of mj} until the end of this section. Assuming the proposition, our improved attack works as follows: for $1 \leq j \leq q$,  we compute a set of $m_j$ from the samples. To avoid dividing by zero, we ignore the samples with $\rho(a) \in \bF_q$ (which happens with probability $1/q$ since $\rho(a)$ is uniformly distributed). We then run a chi-square test on the $m_j$ values. If $j \neq i$, then the distribution should be uniform; if $j = i$, then $P(m_i = s_0) = P(e \in \bF_q)$, which by our assumption is larger than $1/q$. Hence if we plot the histogram of the $m_i$ computed from the samples, we will see a spike at $s_0$. So we could recover $s_0$ as the element with the highest frequency, and output $\rho(s) = s_0 + t_i$. We give the pseudocode of the attack below.

\begin{algorithm}
\caption{Improved chi-square attack on RLWE modulo $\fq$)}
 \label{alg: chi-square}        % give the algorithm a caption
              % and a label for \ref{} commands later in the document
\begin{algorithmic} % enter the algorithmic environment
    \Require  $K$ -- a number field; $R$ -- the ring of integers of $K$; $\fq$ -- a prime ideal in $K$ above $q$ with residue degree 2; $\mathcal{S}$ -- a collection of $M$ RLWE samples;  $\beta > 0$ -- the parameter used for comparing $\chi^2$
values.
    \Ensure a guess of the value $s \pmod{\fq}$, or {\bf NOT-RLWE}, or {\bf INSUFFICIENT-SAMPLES}
    \State Let $\cG \gets \emptyset$.
    \For{$j$ in $1, \ldots, q$}
        \State $\cE_j \gets \emptyset$.
        \For{$a,b$ in $\mathcal{S}$}
            \State $\bar{a}, \bar{b} \gets a \pmod{\fq}, b \pmod{\fq}$.
            \State $m_j \gets \frac{\bar{b} - b - \overline{at_j} + at_j}{\bar{a}-a}.  $
            \State add $m_j$ to $\cE_j$.
        \EndFor

        \State  Run a chi-square test for uniform distribution on $\cE_j$.

        \If{$\chi^2(\cE_j) >  \beta$}
            \State $s_0 := $ the element(s) in $\cE_j$ with highest frequency.
            \State $s \gets s_0 + t_j$, add $s$ to $\cG$.
        \EndIf
    \EndFor
    \If{$\cG = \emptyset$}

        \Return {\bf NOT-RLWE}
    \ElsIf{$\cG = \{s\}$ is a singleton}

        \Return $s$
    \Else

        \Return {\bf INSUFFICIENT-SAMPLES}
    \EndIf

\end{algorithmic}
\end{algorithm}

We analyze the complexity of our improved attack. There are $q$ iterations, each operating on $O(q)$ samples, and reduction of each sample is $O(n)$. So our new attack has complexity $O(nq^2)$.
\subsection{Examples of successful attacks}

To illustrate the idea, we apply our improved attack to the instances in Table 1 of \cite{cryptoeprint:2015:971}. Comparing the last column with the current Table~\ref{tab: attacked}, we see that the runtime has been improved significantly.
\begin{table}[h]
\caption{RLWE instances under our improved attack}
\label{tab: attacked}
\begin{center}
\scalebox{1}{
\begin{tabular}{c|c|c|c|c|c|c}
$n$ & $q$ & $f$ & $r_0$ & no. samples & old runtime (in minutes) & new runtime (in minutes) \\ \hline
 40 & 67 & 2 & 2.51 & 22445 & 209 & 3.5 \\
60 & 197 & 2 & 2.76 & 3940 &  63 & 2.4 \\
 60 & 617 & 2 & 2.76 & 12340 & 8.2 $\times 10^5$ (est.) & 21.3  \\
 80 & 67 & 2 & 2.51 & 3350 & 288.6 & 0.5 \\
 90 & 2003 & 2 & 3.13 & 60090 & 6.6 $\times 10^4$ (est.) &  305 \\
96 & 521  & 2 & 2.76 & 15630 & 4.5 $\times 10^3$ (est.) & 21.7 \\
 100 & 683 & 2 & 2.76 & 20490 &  1.6 $\times 10^4$(est.) & 36.5  \\
 144 & 953 & 2 & 2.51 &  38120 & 342.6 &  114.5
\end{tabular}
}
\end{center}
\end{table}

\subsection{Proof of Proposition~\ref{dist of mj}}

%We will use the following fact from probability theory.
%\begin{Fact}
%Let $X,Y,Z$ be random variables such that $Y$ is independent to $X$ and $Z$. Then for any $x,y,z$, we have
%\[
%    P(X=x,Y=y|Z=z) = P(X=x|Z=z) P(Y=y).
%\]
%\end{Fact}

%\begin{proof}
%LHS = $P(xyz)/P(z)$. RHS = $P(xz)/P(z) \cdot P(y)$.
%\end{proof}

For notational convenience, we let $A_q$ denote the set $\bF_{q^2} \setminus \bF_q$.
\begin{Lemma} \label{lem: prob}
Let the random variable $a$ be uniformly distributed in $A_q$. Suppose $e$ is a random variable with value in $\bF_{q^2}$ independent of $a$. Fix $\delta \in A_q$ and $s_0 \in \bF_q$. Then
$$m_\delta = g_\delta + s_0 + \frac{\bar{e} - e}{\bar{a}-a}$$
is uniformly distributed in $\bF_q$.
Here $$g_\delta = \frac{\overline{a\delta} - a\delta}{\bar{a}-a}.$$
\end{Lemma}

\begin{proof}
Since the uniform distribution is invariant under translation, we may assume $s_0 = 0$. We introduce a new set
$V = \{x \in \bF_{q^2}: \bar{x} = -x\}$. We claim that for any $c,d \in V$ with $c \neq 0$, we have $P(\bar{a} - a = c, \overline{a\delta} - a\delta = d) = \frac{1}{q(q-1)}$.  To prove the claim, note that $V$ is an $\bF_q$-vector space of dimension one, and we have the following $\bF_q$-linear map $f_\delta: \bF_{q^2} \to V^2$.
\[
	f_\delta :  a \mapsto ( \bar{a }  - a, \overline{a\delta} - a\delta).
\]
First we show $f_\delta$ is injective: if $f_\delta(a) = 0$, then $a \in \bF_q$ and thus $a(\bar{\delta} - \delta) = 0$, so $a = 0$. By dimension counting, $f_\delta$ is an isomorphism. Restricting to $A_q$, we see that $f_\delta|_{A_q}$ gives an isomorphism between $A_q$ and $(V \setminus \{0\}) \times V$. This proves the claim.

Let $e' = \frac{\bar{e} - e}{\bar{a}-a}$. For any $z \in \bF_q$,  we have
\begin{align*}
& \quad P(g_\delta + e' = z) \\
& = \sum_{x+y = z} P(g_\delta =x, e' =y) \\
& = \sum_{x+y = z} \sum_{c \in V \setminus \{0\}} P(\bar{a\delta} - a\delta = xc, \bar{e} - e = yc, \bar{a} - a = c) \\
& = \sum_{x+y = z, c\in V \setminus \{0\}} P(\bar{a\delta} - a\delta = xc , \bar{a} - a = c)  P(\bar{e} - e = yc) \\
& = \frac{1}{q(q-1)}\sum_{y \in \bF_q , c\in V \setminus \{0\}} P(\bar{e} - e = yc) \\
& = \frac{1}{q(q-1)}  \cdot (q-1) \sum_{c' \in V} P(\bar{e} - e = c') \\
& = \frac{1}{q}.
\end{align*}
\end{proof}

\begin{proof}[of Proposition~\ref{dist of mj}]
The second claim follows directly from (1). For the first claim, let $\delta = t_i - t_j$. Then
$m_j \sim g_\delta + s_0 + \frac{\bar{e} - e}{\bar{a}-a},$ where $g_\delta = \frac{\overline{a\delta} - a\delta}{\bar{a}-a}$.
Now the first claim is precisely Lemma~\ref{lem: prob}.
\end{proof}

\section{Infinite family of vulnerable Galois RLWE instances}
\label{section: family}

Recall that a number field $K$ of degree $n$ is {\it Galois} if it has exactly $n$ automorphisms.
In this section, we describe Galois number fields which are vulnerable to the attack outlined in Section~\ref{subsec: attack} 
In contrast to the vulnerable instances found by computer search in Section 5 of \cite{cryptoeprint:2015:971}, in this section we explicitly construct infinite families of such fields with flexible parameters.
Furthermore, the attacks of \cite{cryptoeprint:2015:971} were successful only on instances where the size of the distribution (in the form of the scaled standard deviation) is a small constant, where as in this paper the scaled standard deviation parameter can be taken to be $o(|d|^{1/4})$, where $d$ is an integer parameter and can go to infinity.

%Notice that the attack requires looping through all guesses $g$ in $F$.  In the next section, we will improve this attack to avoid such a large loop.

To set up, let $p$ be an odd prime and let $d > 1$ be a squarefree integer such that $d$ is  coprime to $p$ and $d \equiv 2,3 \mod{4}$. We choose an odd prime $q$ such that
\begin{enumerate}
\item[(1)] $q \equiv 1 \pmod{p}$.
\item[(2)] $\left( \frac{d}{q} \right) = -1$  (equivalently, the prime $q$ is inert in $\bQ(\sqrt{d})$).
\end{enumerate}

\begin{remark}
Fix a pair  $(p,d)$ that satisfies the conditions described above.  By quadratic reciprocity, condition (2) on $q$ above is a congruence condition modulo $4d$.  So by Dirichlet's theorem on primes in  arithmetic progressions, there exists infinitely many primes $q$ satisfying both (1) and (2).
\end{remark}

Let $M = \bQ(\zeta_p)$ be the $p$-th cyclotomic field and $L = \bQ(\sqrt{d})$. Let $K = M  \cdot L$ be the composite field and let $\cO_K$ denote its ring of integers.  

%For an element $v \in K$, we let $||v||$ denote its 2-norm under the embedding, i.e., $||v|| := ||\iota(v)||_2$. We will call this the {\it embedding length} of $v$. Our goal in this section is to prove:

\begin{theorem} \label{thm in sec2}
Let $K$ and $q$ be as above, and $R_q$ defined as in the preliminaries in terms of $K$ and $q$.  Suppose $\fq$ is a prime ideal in $K$ lying over $q$. We consider the reduction map $\rho  : R/qR \to R/\fq R \cong \mathbb{F}_{q^f}$, where $f$ is the residue degree. Suppose $\cD$ is the RLWE error distribution with error width $r$ such that $r < 2\sqrt{\pi d}$. Let
 \[
	\beta = \min \left\{ \left(\frac{\sqrt{4\pi e d}}{r} e^{-\frac{2\pi d}{r^2}} \right)^n, 1 \right\}.
\]
Then, for $x \in R_q$ drawn according to $\cD$, we have $\rho(x) \in \mathbb{F}_q$ with probability at least $1-\beta$. 
\end{theorem}

\begin{example}
As a sample application of the theorem, we take $d = 4871, r = 68.17$ and $p = 43$. Then we computed $\beta = 0.11\ldots$. So if $x \in R_q$ is drawn from the error distribution, then $\rho(x) \in \bF_q$ with probability at least 0.88. 
\end{example}

\begin{Lemma} \label{field theory}Under the notation above, we have \\
(1) $K/\bQ$ is a Galois extension. \\
(2) $[K: \bQ] = [M:\bQ][L:\bQ] = 2(p-1)$. \\
(3) The prime $q$ has residue degree 2 in $K$.  \\
(4) $\cO_K = \cO_M \cdot \cO_L = \bZ[\zeta_p, \sqrt{d}]$. \\
(5) $|\disc(\cO_K)| = p^{2(p-2)}(4d)^{(p-1)}$.
\end{Lemma}

\begin{proof}
(1) follows from the fact that $K$ is a composition of Galois extensions $M$ and $L$; (2) is equivalent to $M \cap L = \bQ$, which holds because $L/\bQ$ is unramified away from primes dividing $2d$ and $M/\bQ$ is unramified away from $p$;
for (3), note that our assumptions imply that $q$ splits completely in $M$ and is inert in $L$, hence the claim. The claims (4) and (5) follow directly from \cite[II. Theorem 12]{marcus1977number}, and the fact that $\disc(\cO_M) = p^{p-2}$ and $\disc(\cO_L) = 4d$ are coprime.
\end{proof}

The following lemma is a standard upper bound on the Euclidean lengths of samples from discrete Gaussians. It can be deduced directly from \cite[Lemma 2.10]{micciancio2007worst}.
\begin{Lemma}
\label{lem: bound}
Suppose $\Lambda \subseteq \bR^n$ is a lattice. Let $D_{\Lambda,r}$ denote the discrete Gaussian over $\Lambda$ of width $r$. Suppose $c$ is a positive constant such that $c > \frac{r}{\sqrt{2\pi}}$. Let $v$ be a sample from $D_{\Lambda, r}$. Then 
\[
	\Prob(||v||_2 > c \sqrt{n}) \leq C_{c/r}^n, 
\]
where $C_s = s \sqrt{2 \pi e} \cdot e^{-\pi s^2}$.
\end{Lemma}
\begin{proof}[of Theorem]
Part (3) of Lemma~\ref{field theory} implies that
\begin{equation}
	1, \zeta_p, \ldots, \zeta_p^{p-2};   \sqrt{d}, \ldots, \zeta_p^{p-2} \sqrt{d} \tag{*}
\end{equation}
is an integral basis of $R = \cO_K$. By our assumptions, we have $R/\fq R \cong \bF_{q^2}$, the finite field of $q^2$ elements. Under the map $\rho$, the first $(p-1)$ elements of the basis reduce to $\bF_q$, and the rest reduce to the complement $\bF_{q^2} \setminus \bF_q$, because $d$ is not a square modulo $q$.

Let $n = p - 1$ be the degree of $M$ over $\bQ$. Then the extension $K/\bQ$ has degree $2n$. We denote the elements in (*) by $v_1, \ldots, v_n$ and $w_1, \ldots, w_n$. Then $||\iota(v_i)|| = \sqrt{2n}$, while $||\iota(w_i)|| = \sqrt{2nd}$.  
We compute the root volume $c := \left( vol(R) \right)^{1/n}$. It is a general fact that $vol(R) = |\disc(R)|^{\frac{1}{2}}$, so we have
\[
	c = |\disc(R)|^{\frac{1}{2n}} = \sqrt{2} p^{\frac{p-2}{2(p-1)}}  d^{\frac{1}{4}}.
\]
So when $d \gg p$, we have $||v_i|| \ll c \ll ||w_i||$. We have a decomposition $R = V \oplus W$, where $V$ and $W$ are free abelian groups with bases $v_1, \ldots, v_n$ and $w_1, \ldots, w_n$, respectively. The embeddings of $V$ and $W$ are orthogonal subspaces, because $\Tr(v_i \bar{w_j}) = 0$ for all $i, j$.  For any element $e \in R$, we can write $e = e_1 + e_2 \sqrt{d}$ where $e_1, e_2$ are elements of $\bZ[\zeta_p]$, and it follows that $||e||^2 = ||e_1||^2 + d||e_2||^2$.
In particular, if $e_2 \neq 0$, then $||e|| \geq \sqrt{2nd}$.

By applying Lemma~\ref{lem: bound} with $c = \sqrt{2d}$, the assumptions in the statement of our theorem imply that the probability that the discrete Gaussian $D_{\iota(R), r}$ will output a sample with $e_2 \neq 0$ is less than $\beta$. So the statement of theorem follows, since $e_2 = 0$ implies $\rho(e) \in \bF_q$, i.e., the image of $e$ lies in the prime subfield. 
\end{proof}

%Hence this can be detected by mapping the errors to $\bF_{q^2}$ via $\rho$ and counting the images that happen to lie in $\bF_q$.
Therefore, we can specialize the general attack in this situation as follows. Given a set $S$ of samples $(a,b) \in (R/qR)^2$, we loop through all $q^2$ possible guesses $g$ of the value $s \mod{\fq}$ and compute $e_g = \rho(b) - g \rho(a)$. We then perform a chi-square test on the set $\{e_g: (a,b) \in S \}$, using two bins $\bF_q$ and $\bF_{q^2} \setminus \bF_q$. If the samples are not taken from the RLWE distribution, or if the guess is incorrect, we expect to obtain uniform distributions; for the correct guess, we have $e_g = \rho(e)$, and by the above analysis, if the error parameter $r_0$ is sufficiently small, then the chi-square test might detect non-uniformness, since the portion of elements that lie in $\bF_q$ might be larger than $1/q$.

The theoretical time complexity of our attack is $O(nq^3)$: the loop runs through $q^2$ possible guesses. In each passing of the loop, the number of samples we need for the chi-square test is $O(q)$, and the complexity of computing the map $\rho$ on one sample is $O(n)$. Note that using the techniques in Section 3 of this paper, we could reduce the 
complexity to $O(nq^2)$. 

\begin{remark}
It is easy to verify that if a triple $(p,q,d)$ satisfies our assumptions, then so does $(p,q, d+4kq)$ for any integer $k$, as long as $d+4kq$ is square free. This shows one infinite family of Galois fields vulnerable to our attack.
\end{remark}

\subsection{Examples}

Table~\ref{tab: new attacked} records some of the successful attacks we performed on the instances described previously.  In each row of Table~\ref{tab: new attacked}, the degree of the number field is $2(p-1)$. Note that the 
runtimes are computed based on the improved version of the attack described in Section 3 of this paper. Also, 
by varying the parameters $p$ and $d$, we can find vulnerable instances with $r_0 \to \infty$. For example, any $r_0 = o(d^{1/4}/\sqrt{p})$ will suffice. 

\begin{remark}
From Table~\ref{tab: new attacked}, we see that the the attack in practice seems to work better (i.e., we can attack larger width $r$) than what is predicted in Theorem~\ref{thm in sec2}. As a possible explanation, we remark that in proving the theorem we bounded the probability of $e_2 = 0$ from below. However, the condition $e_2 = 0$ is sufficient but not necessary for $\rho(e)$ to lie in $\bF_q$, so our estimation may be a very loose one. 
\end{remark}

\begin{table}[h!]
\caption{New vulnerable Galois RLWE instances}
\label{tab: new attacked}
\begin{center}
\begin{tabular}{c|c|c|c|c|c|c}
$p$ & $d$ & $q$ & $r_0$ & $r$ & no. samples & runtime (in seconds) \\ \hline
31 & 4967 & 311 & 8.94& 592.94 & 3110 & 144.92 \\
43 & 4871 & 173 & 8.97 & 694.94 & 1730 & 6.44 \\
61 & 4643 & 367 & 8.84  &815.11 & 3670 & 205.28 \\
83 & 4903 & 167  & 8.94 & 963.84 & 1670  &  5.74 \\
103 & 4951 & 619 & 8.94 & 1076.32 & 6190 & 579.77\\
109 & 4919 & 1091 & 8.94 & 1105.44 & 10910 & 1818.82 \\
151 &100447 &907 & 14.08 & 4356.02 & 9070 &1394.18 \\
181 & 100267 & 1087 & 14.11 & 4777.17 & 10870 & 1973.47 \\
\end{tabular}
\end{center}
\end{table}

\subsection{Remarks on other possible attacks}

First, we note that the instances we found in this section are not directly attackable using linear algebra, as in the recent paper \cite{castryck2016provably}. The reason is that although the last $n/2$-coordinates of the error $e$ under the basis (*) are small integers, they are nonzero most of the time, so it is not clear how one can extract exact linear equations from the samples. On the other hand, note that for linear equations with small errors, there is the attack on the search RLWE problem proposed by Arora and Ge. However, the attack requires $O(n^{d - 1})$ samples and solving a linear system in $O(n^{d})$ variables. Here $d$ is the width of the discrete error: for example, if the error can take values $0,1,2,-1,-2$, then $d = 5$. Thus the attack of Arora and Ge becomes impractical when $n$ is larger than $10^2$ and $d \geq 5$, say. In contrast, the complexity of our attack depends linearly on $n $ and quadratically on $q$. In particular, it does not depend on the error size (although the success rate does depend on the error size).

\section{Security of 2-power cyclotomic rings with unramified moduli}
\label{sec: cyclo-secure}

In this section we provide some numerical evidence that for 2-power cyclotomic rings, the image of a fairly narrow RLWE error distribution modulo an unramified prime ideal $\fq$ of residue degree one or two is practically indistinguishable from uniform, implying that the 2-power cyclotomic rings are protected against the family of attacks in this paper. 

We restrict ourselves to 2-power cyclotomic rings because the geometry is simple, namely the discrete Gaussian distribution $D_{\iota(R),\sqrt{n} r}$ over the ring is equivalent to a PLWE distribution, where each coefficient of the error is sampled independently from a discrete Gaussian $D_{\bZ,r}$ over the integers.

To further aid the analysis, we make another simplifying assumption by replacing $D_{\bZ,r}$ in the PLWE distribution described above by a  ``shifted binomial distribution''. This allows a closed form formula for a bound on the statistical distance, and hence eases the analysis. 

%We also generate the actual RLWE samples, run our chi-square attack, and confirm that the errors modulo $\fq$ are indeed uniform.

Let $m = 2^d$ for some integer $d \geq 1$ and let $K = \bQ(\zeta_m)$ be the $m$-th cyclotomic field, with degree $n = m/2$. Let $q$ be a prime such that $q \equiv 1\pmod{m}$. Finally, let $\fq$ be a prime ideal above $q$.

\iffalse
First, we introduce the PLWE error distribution on cyclotomic fields, which is commonly used in practice for homomorphic encryption schemes as a substitute for the RLWE error distribution. Let $n = \varphi(m)$ be the degree of $K$.
\begin{definition}
Let $\tau > 0$. A sample from the {\it PLWE distribution} $P_{m,\tau}$ is
\[
    e = \sum_{i=0}^{n-1} e_i \zeta_m^i,
\]
where the $e_i$ are sampled independently from the discrete Gaussian $D_{\bZ,\tau}$.
\end{definition}
\fi
Now we introduce a class of ``shifted binomial distributions''.
\begin{definition}
For an even integer $k \geq 2$, let $\cV_{k}$ denote the distribution over $\bZ$ such that for every $t \in \bZ$,

$$\prob(\cV_{k} = t) =  \begin{cases} \frac{1}{2^k}{k \choose t+\frac{k}{2}} &\mbox{if } |t| \leq \frac{k}{2} \\
0 & \mbox{otherwise}  \end{cases}$$

\end{definition}
We will abuse notation and also use $\cV_{k}$ to denote the reduced  distribution $\cV_{k} \pmod {q}$ over $\bF_q$, and let $\nu_{k}$ denote its probability density function. Figure~\ref{fig: v8} shows a plot of $\nu_8$.
\begin{figure}[h!]
\centering
\includegraphics[width = 0.5\textwidth]{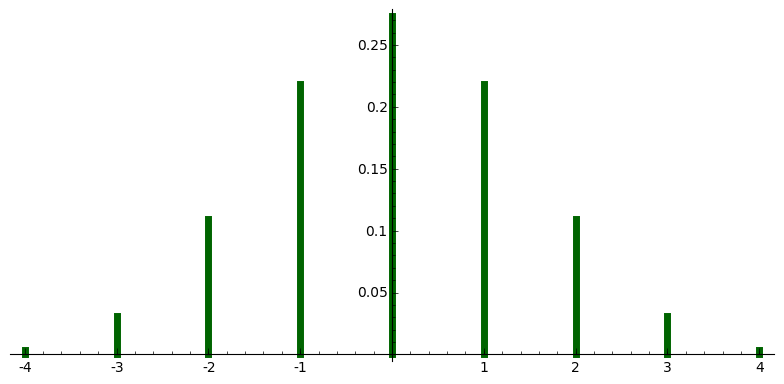}
\caption{Probability density function of $\cV_{8}$}
\label{fig: v8}
\end{figure}

\begin{definition}
\label{def: helper dist}
Let $k \geq 2$ be an even integer. Then a sample from the distribution $P_{m,k}$ is
\[
    e = \sum_{i=0}^{n-1} e_i \zeta_m^{i},
\]
where the coefficients $e_i$ are sampled independently from $\cV_k$.
\end{definition}

\subsection{Bounding the Distance from Uniform}

We recall the definition and key properties of Fourier transform over finite fields.
Suppose $f$ is a real-valued function on $\bF_q$. The {\it Fourier transform} of $f$ is defined as
\[
        \widehat{f}(y) = \sum_{a \in \bF_q} f(a) \overline{\chi_y(a)},
\]
where $\chi_y(a) := e^{2 \pi i ay/q}$.

Let $u$ denote the probability density function of the uniform distribution over $\bF_q$, that is $u(a) = \frac{1}{q}$ for all $a \in \bF_q$. Let $\delta$ denote the characteristic function of the
one-point set $\{0\} \subseteq \bF_q$. Recall that the convolution of two functions $f,g: \bF_q \to \bR$ is
defined as $(f  \ast g ) (a) = \sum_{b \in \bF_q} f(a-b)g(b)$. We list without proof some basic properties of the Fourier transform.
\begin{enumerate}
\item $\widehat{\delta} = qu$; $\widehat{u} = \delta$.
\item $\widehat{f \ast g} = \widehat{f} \cdot \widehat{g}$.
\item $f(a) = \frac{1}{q} \sum_{y \in \bF_q} \widehat{f}(y)\chi_y(a)$ (the Fourier inversion formula).
%\item Plancherel's formula states that
%$||f||_2 = \frac{1}{q} ||\widehat{f}||_2$.
%\item $\widehat{f(a - \lambda)}(s) =  \bar{\chi_s(\lambda)} f(s)$.
\end{enumerate}

%Suppose $f,g$ are the probability density functions of two random variables $F,G$ with value in $\bF_q$. Let $h$ denote the density function of the sum $H =F+G$.
The following is a standard result.

\begin{Lemma}
Suppose  $F$ and $G$ are independent random variables with values in $\bF_q$, having probability density functions $f$ and $g$.  Then the density function of $F+G$  is equal to $f \ast g$.
In general, suppose $F_1, \ldots, F_n$ are mutually independent random variables in $\bF_q$, with probability density functions $f_1, \ldots, f_n$. Let $f$ denote the density function of the sum $F = \sum F_i$, then $f = f_1 \ast \cdots \ast f_n$.
\end{Lemma}

%\begin{proof}
%We prove the first claim. For any $a \in \bF_q$,
%\begin{align*}
%\prob(F+G = a) &= \sum_{b \in \bF_q} \prob(F = a-b, G =b) \\
%& = \sum_{b \in \bF_q} \prob(F = a-b)\prob(G =b).  \qquad \mbox{(since $F,G$ are independent)} \\
%& = (f \ast g)(a).
%\end{align*}
%The general case follows from an induction on $n$.
%\qed \end{proof}

The Fourier transform of $\nu_k$ has a nice closed-form formula, as below.
\begin{Lemma}
\label{lem: transform1}
For all even integers $k \geq 2$, $\widehat{\nu_k}(y)  = \cos \left(\frac{\pi y}{q}\right)^k.$
\end{Lemma}

\begin{proof}  We have
\begin{align*}
2^k \cdot \widehat{\nu_k}(y) &= \sum_{m \in \bZ/q\bZ} \left( \sum_{a \in \bZ: |aq+m| \leq k/2} {k \choose aq+m + \frac{k}{2}} \right) e^{-2 \pi i ym/q}  \\
& =\sum_{m = -\frac{k}{2}}^{\frac{k}{2}} {k \choose m+\frac{k}{2}} e^{2\pi i ym/q}  \\
&= e^{-\pi i yk/q} \sum_{m' = 0}^{k} {k \choose m'} e^{2\pi i ym'/q} \\
& =  e^{-\pi i yk/q} (1+ e^{2 \pi i y/q})^k  = (2 \cos(\pi y/q))^k.
\end{align*}
Dividing both sides by $2^k$ gives the result.
 \end{proof}

% $P_{m,\tau} \pmod{\fq}$

Next, we concentrate on the ``reduced distribution" $P_{m,k} \pmod {\fq}$. Note that there is a one-to-one correspondence between primitive $m$-th roots of unity in $\bF_q$ and the prime ideals above $q$ in $\bQ(\zeta_m)$. Let $\alpha$ be the root corresponding to our choice of $\fq$. Then a sample from $P_{m,k} \pmod {\fq}$ is
of the form

%$P_{m, \tau} \pmod{\fq}$ (resp.

$$e_\alpha = \sum_{i=0}^{n-1} \alpha^i e_i \pmod {q},$$
where the coordinates $e_i$ are independently sampled from $\cV_k$. We abuse notations and use $e_\alpha$ to denote its own probability density function.

\begin{Lemma}
\label{lem: transform2}
\[
    \widehat{e_\alpha}(y) = \prod_{i=0}^{n-1} \cos \left(\frac{ \alpha^i \pi y}{q} \right)^k.
\]
\end{Lemma}

\begin{proof}
This follows directly from Lemma~\ref{lem: transform1} and the independence of the coordinates $e_i$.
 \end{proof}

\begin{Lemma} \label{prop: bound}
Let $f: \bF_q \to \bR$ be a function such that $\sum_{a \in \bF_q} f(a) = 1$. Then for all $a \in \bF_q$, the following holds. 
\begin{equation} \label{eq: secure}
    |f(a) -  1/q| \leq \frac{1}{q}  \sum_{y \in \bF_q, y \neq 0}  |\hat{f}(y)|.
\end{equation}
\end{Lemma}

\begin{proof} For all $a \in \bF_q$,
\begin{align*}
    f(a) - 1/q &= f(a) - u(a) \\
    & = \frac{1}{q} \sum_{y \in \bF_q} (\hat{f}(y) - \widehat{u}(y) )\chi_y(a) \\
& = \frac{1}{q} \sum_{y \in \bF_q} (\hat{f}(y)  - \delta(y) )\chi_y(a) \\
& = \frac{1}{q} \sum_{y \in \bF_q, y \neq 0} \hat{f}(y)  \chi_y(a).  \qquad \mbox{(since $\hat{f}(0) = 1$)}
\end{align*}
Now the result follows from taking absolute values on both sides, and noting that $|\chi_y(a)| \leq 1$ for all $a$ and all $y$.
 \end{proof}

Taking $f = e_\alpha$ in Lemma~\ref{prop: bound}, we immediately obtain
\begin{theorem} \label{cor: stat dist}
The statistical distance between $e_\alpha$ and $u$ satisfies

%$$d(e_\alpha,u) \leq \frac{1}{2}  \sum_{y \in \bF_q, y \neq 0}  |\widehat{e_\alpha}(y)|.$$
%Similarly,
\begin{equation} \label{distance}
\Delta(e_\alpha,u) \leq \frac{1}{2}  \sum_{y \in \bF_q, y \neq 0}  |\widehat{e_\alpha}(y)|.
\end{equation}
\end{theorem}

Now let $\epsilon(m,q,k,\alpha)$ denote the right hand side of (\ref{distance}), i.e.,
\[
    \epsilon(m,q,k, \alpha) = \frac{1}{2}\sum_{y \in \bF_q, y \neq 0} \prod_{i=0}^{n-1} \cos \left(\frac{ \alpha^i \pi y}{q} \right)^k.
\]
To take into account all prime ideals above $q$, we let $\alpha$ run through all primitive $m$-th roots of unity in $\bF_q$ and define
$$\epsilon(m,q,k) := \max \{\epsilon(m,q,k,\alpha): \alpha \mbox{ has order } m \mbox{ in } (\bF_q)^*\}.$$
If $\epsilon(m,q,k)$ is negligibly small, then the distribution $P_{m,k} \pmod {\fq}$ will be computationally indistinguishable from uniform.  We will prove the following theorem.

\begin{theorem} \label{cyclo secure}
Let $q,m$ be positive integers such that $q$ is a prime, $m$ is a power of 2, $q \equiv 1\mod{m}$ and  $q < m^2$. Let $\beta = \frac{1 + \frac{\sqrt{q}}{m}}{2}$; then $0 < \beta < 1$ and
\[
	\epsilon(m,q,k)  \leq \frac{q-1}{2} \beta^{\frac{km}{4}}.
\]
\end{theorem}

In particular, if $\beta^{k/4} < \frac{1}{2}$, then the theorem says that $\epsilon(m,q,k)  = O(q2^{-m})$ as $m \to \infty$.

\begin{corollary}
        The statistical distance between $P_{m,k}$ modulo $\mathfrak{q}$ and a uniform distribution is bounded above, independently of the choice of $\mathfrak{q}$ above $q$, by
\begin{equation*}
        \frac{q-1}{2} \left(  \frac{1 + \frac{\sqrt{q}}{m}}{2}\right)^{\frac{km}{4}}.
\end{equation*}
\end{corollary}

To prepare proving the theorem, we set up some notations of Shparlinski in \cite{shparlinski1995some}. Let $\Omega = (\omega_j)_{j=1}^{\infty}$
be a sequence of real numbers and let $m$ be a positive integer. We define the following quantities:
\begin{itemize}
\item $L_\Omega(m) =  \prod_{j=1}^{m} ( 1 - \exp(2\pi i \omega_j))$ \\
\item $S_\Omega(m) =  \sum_{j=1}^{m}  \exp(2\pi i \omega_j)$.
\end{itemize}

The following lemma is a special case of \cite[Theorem 2.4]{shparlinski1995some}.
\begin{Lemma} \label{igor}
$$|L_\Omega(m)| \leq 2^{m/2} (1 + |S_\Omega(m)|/m)^{m/2}.$$
\end{Lemma}

\begin{proof}[of Theorem~\ref{cyclo secure}] We specialize the above discussion to our situation, where $m$ is a power of 2 and $n = m/2$.  We fix $\omega_k = \frac{\alpha^{k-1} y}{q} + 1/2$,
where we abuse notations and let $\alpha$ denote a lift of $\alpha \in \bF_q$ to $\bZ$.

\begin{Lemma} \label{myself}
We have 
\[
	|L_\Omega(n)|= 2^n \left| \prod_{j=0}^{n-1} \cos \left(\frac{ \alpha^j \pi y}{q} \right) \right|
\]
and $|L_\Omega(m)|= |L_\Omega(n)|^2$. 
\end{Lemma}

\begin{proof}
 We have $L_\Omega(n) = \prod_{j=1}^n ( 1 - e^{2\pi i ( \alpha^{j-1} y /q + 1/2)})  = \prod_{j=0}^{n-1} ( 1 + e^{2\pi i \alpha^{j} y/q})$. So
 \(
  |L_\Omega(n)| = \prod_{j=0}^{n-1} \big|e^{- \pi i \alpha^j y/q}+ e^{\pi i \alpha^jy/q)}\big|
  = \prod_{j=0}^{n-1} 2\,\big|\Re(e^{\pi i \alpha^j y /q})\big|
 \), which is equal to
 \( 2^n \big| \prod_{j=0}^{n-1} \cos(\alpha^j \pi y/q)\big|
 \).
 Similarly, $|L_\Omega(m)| = 2^m | \prod_{j=0}^{m-1} \cos(\alpha^j \pi y/q)|$. Since $\alpha^n \equiv -1 \mod{q}$ we have
 $\cos(\alpha^{j+n} \pi y/q) = \cos( \alpha^j \pi y/q)$ for $0 \leq j \leq n-1$. The claim now follows.
\end{proof}

On the other hand, we have $S_\Omega(m) =   -  \sum_{j=0}^{m-1} \exp \left(\frac{ 2\pi  i \alpha^j y}{q} \right)$, and standard bound on Gauss sums says that $|S_\Omega(m)| \leq q^{1/2}$.  Now combining Lemma~\ref{igor} and Lemma~\ref{myself}, we get
$$\left| \prod_{i=0}^{n-1} \cos \left(\frac{ \alpha^i \pi y}{q} \right) \right|  \leq \beta^{n/2}$$
for $\beta$ as defined in the statement of the theorem and for any nonzero $y \in \bF_q$. 
Our result in the theorem now follows from taking both sides to $k$-th power and summing over $y$.
\end{proof}

\subsection{Numerical Distance from Uniform}

We have computed $\epsilon(m,q,k)$ for various choices of parameters.  Smaller values of $\epsilon$ imply that the error distribution looks more uniform when transferred to $R/\mathfrak{q}$, rendering the instance of RLWE invulnerable to the attacks in \cite{cryptoeprint:2015:971}.

The data in Table~\ref{tab: deg1} shows that when $n \geq 100$ and the size of the modulus $q$ is polynomial in $n$, the statistical distances between $P_{m,k} \pmod{\fq}$ and the uniform distribution are negligibly small. Also, note that we fixed $k= 2$, and the epsilon values becomes even smaller when $k$ increases.

For each instance in Table~\ref{tab: deg1}, we also  generated the actual RLWE samples (where we fixed $r_0 = \sqrt{2\pi}$) and ran the chi-square attack of  \cite{cryptoeprint:2015:971} using confidence level $\alpha = 0.99$. The column labeled ``$\chi^2$'' contains the $\chi^2$ values we obtained, and the column labeled ``uniform?'' indicates whether the reduced errors are uniform. We can see from the data how the practical situation agrees with our analysis on the approximated distributions.

\iffalse
%$-[\log_2(\epsilon(m,q, 1))]$
38
51
79
203
248
253
410
511
\fi
\begin{table}[h!]
\caption{Values of $\epsilon(m,q,2)$ and the $\chi^2$ values}
\label{tab: deg1}
\begin{center}
\begin{tabular}{c|c|c|c|c}
$m$ ($n =m/2$) & $q$ & $-[\log_2(\epsilon(m,q, 2))]$ & $\chi^2$ & uniform? \\
\hline
64  & 193 & 40 & 167.6  & yes \\
128 & 1153 & 97 & 1125.6 & yes\\

%96 & 32 & 193 & $35$ & 231.6 & yes \\
%55 & 40 & 331  & $44$ & 308.8 & yes \\
%160 & 64 & 641 & 55 & 658.0 & yes \\
%101 & 100 & 1213 & $177$ &  1254.4 & yes \\
%145 & 112  &19163 & $176$  & 176\\
%244 & 120 & 1709 & 230 &  1721.2 & yes \\
256 & 3329 & $194$ &  3350.0 & yes \\
%256 & 128 & 14081 & $208$ & 240 \\
%197 & 196 & 3547 & $337$ &  3475.2 & yes \\
512  &10753 & 431 & 10732.8 & yes\\
% 512 & 256 &19457 & 414 & 512 &
\end{tabular}
\end{center}
\end{table}

It is possible to generalize our discussion in this section to primes of arbitrary residue degree $f$, in which case the Fourier analysis will be performed over the field $\bF_{q^f}$. The only change in the definitions would be $\chi_y(a) = e^{ \frac{2 \pi i \Tr(a y)}{q}}$. Here $\Tr: \bF_{q^f} \to \bF_q$ is the trace function. Similarly, we have
\[
    \widehat{e'_\alpha}(y) = \prod_{i=1}^{n} \cos \left(\frac{ \pi \Tr(\alpha^i y) }{q} \right)^k.
\]
Table~\ref{tab: deg2} contains some data for primes of degree two.

\begin{table}[h]
\caption{Values of $\epsilon(m,q,2)$ for primes of degree two}
\begin{center}
\begin{tabular}{c|c|c} \label{tab: deg2}
$m$ ($n = m/2$) & $q$ & $-[\log_2(\epsilon(m,q,2))]$ \\
\hline
64 &383 & 31 \\
128 & 1151 & 54 \\
256 & 1279 & 159 \\
%63 & 36 & 881 & 33 \\
%55 & 40 &109 & 48 \\
%53 & 52  & 211 & 61 \\
512  & 5583 & 341
\end{tabular}
\end{center}
\end{table}

\bibliographystyle{splncs03}
\bibliography{sac-rlwe}

\end{document}